\documentclass[runningheads]{llncs}
\usepackage[T1]{fontenc}
\usepackage{graphicx}
\usepackage[a4paper, left=25mm, right=25mm, top=25mm, bottom=25mm]{geometry}
\usepackage[style=ext-numeric,citestyle=numeric-comp,sorting=nyt,sortcites, backend=biber, giveninits=true, maxnames=99]{biblatex}
\DeclareFieldFormat[article,periodical]{volume}{\mkbibbold{#1}}
\DeclareFieldFormat[article,periodical]{number}{\mkbibparens{#1}}
\renewbibmacro{in:}{%
  \ifentrytype{article}
    {}
    {\printtext{\bibstring{in}\intitlepunct}}}

\DeclareFieldFormat{issuedate}{#1}
\renewcommand{\jourvoldelim}{\addcomma\space}

\renewbibmacro*{journal+issuetitle}{%
  \usebibmacro{journal}%
  \setunit*{\jourvoldelim}%
  \iffieldundef{series}
    {}
    {\setunit*{\jourserdelim}%
     \printfield{series}%
     \setunit{\servoldelim}}%
  \usebibmacro{volume+number+eid}%
  \setunit{\addcolon\space}%
  \usebibmacro{issue}%
  \setunit{\bibpagespunct}%
  \printfield{pages}%
  \setunit{\space}%
  \usebibmacro{issue+date}%
  \newunit}

\renewbibmacro*{note+pages}{%
  \printfield{note}%
  \newunit}

\DeclareFieldFormat[article,periodical]{date}{\mkbibparens{#1}}

\renewbibmacro*{publisher+location+date}{%
  \printlist{publisher}%
  \iflistundef{location}
    {\setunit*{\addcomma\space}}
    {\setunit*{\addcolon\space}}%
  \printlist{location}%
  \setunit*{\addcomma\space}%
  \usebibmacro{date}%
  \newunit}

\AtEveryBibitem{\clearfield{month}}
\AtEveryBibitem{\clearfield{day}}

\AtEveryBibitem{
  \ifentrytype{online}{}{
    \ifentrytype{thesis}{}{
      \ifentrytype{unpublished}{}{ 
        \clearfield{url}
      }
    }
  }
  \clearfield{language}
}

\DeclareSourcemap{
  \maps[datatype=bibtex]{
    \map[overwrite=true]{
      \step[fieldset=language, null]
      \step[fieldset=address, null]
     \step[fieldset=pagetotal, null]
    }
  }
}

\DeclareSourcemap{
    \maps[datatype=bibtex]{
      \map{
        \step[fieldsource=eprint,final]
        \step[fieldset=url,null]
        \step[fieldset=doi,null]
      }  
    }
  }

\AtEveryBibitem{\clearfield{pubstate}} 
\AtEveryBibitem{\clearfield{eprintclass}} 
\AtEveryBibitem{\clearfield{version}} 

\usepackage[unicode=true, pdfusetitle, colorlinks=true, allcolors=blue, bookmarks=false]{hyperref}
\usepackage[hyphenbreaks]{breakurl}
\usepackage{xurl}
\hypersetup{breaklinks=true}
\usepackage{color}

\usepackage{amssymb,amsfonts,stmaryrd,mathtools,mathrsfs}
\usepackage{tikz}
\usetikzlibrary{cd}
\usepackage{lipsum}

\newcommand*{\ZZ}{\mathbb{Z}}

\newcommand*{\RR}{\mathbb{R}}

\newcommand\restr[2]{{
  \left.\kern-\nulldelimiterspace 
  #1 
  \right|_{#2} 
}}

\newcommand{\T}{{\mathsf T}}
\newcommand{\cT}{\T^{\ast}}
\newcommand*{\dd}{\mathrm{d}}
\newcommand*{\contr}[1]{\iota_{#1}}
\newcommand*{\liedv}[1]{\mathcal{L}_{#1}}
\newcommand*{\Reeb}{\mathcal{R}}
\newcommand*{\X}{\mathfrak{X}}
\newcommand*{\lvf}{\nabla} 
\newcommand{\Cinfty}{\mathscr{C}^\infty}
\newcommand*{\parder}[2]{\frac{\partial#1}{\partial #2}}
\newcommand{\parderr}[3]{\frac{\partial^2 #1}{\partial #2\partial #3}}

\DeclareMathOperator{\Sec}{Sec}

\let\oldemph\emph

\mathtoolsset{showonlyrefs} 

\usepackage[colorinlistoftodos]{todonotes} 

\begin{document}
\title{\sffamily Homogeneous bi-Hamiltonian structures\\ \sffamily and integrable contact systems}
%
\titlerunning{Homogeneous bi-Hamiltonian structures and integrable contact systems}
%
\author{Leonardo Colombo\inst{1}\orcidID{0000-0001-6493-6113} \and \\
Manuel de León\inst{2,3}\orcidID{0000-0002-8028-2348} \and \\
María Emma Eyrea Irazú\inst{4}\orcidID{{0000-0001-5535-9959}} \and \\
Asier López-Gordón\inst{5}\orcidID{0000-0002-9620-9647} }
\authorrunning{L.~Colombo, M.~de León, M.~E.~Eyrea Irazú, and A.~López-Gordón}
%
\institute{Centre for  Automation and Robotics, Spanish Research Council, Arganda del Rey, Spain. 
\email{leonardo.colombo@car.upm-csic.es}
\and
Institute of Mathematical Sciences and Spanish Royal Academy of Sciences, Madrid, Spain. \\
\email{mdeleon@icmat.es}
\and
La Plata Mathematics Center, UNLP, Buenos Aires, Argentina. \\
\email{maeemma@mate.unlp.edu.ar}
\and
Institute of Mathematics, Polish Academy of Sciences, Warsaw, Poland. \\
\email{alopez-gordon@impan.pl}
}
\maketitle              
\begin{abstract}


Bi-Hamiltonian structures can be utilised to compute a maximal set of functions in involution for certain integrable systems, given by the eigenvalues of the recursion operator relating both Poisson structures. We show that the recursion operator relating two compatible Jacobi structures cannot produce a maximal set of functions in involution. However, as we illustrate with an example, bi-Hamiltonian structures can still be used to obtain a maximal set of functions in involution on a contact manifold, at the cost of symplectisation.
\end{abstract}
\section{Introduction}

A $2n$-dimensional Hamiltonian system $(M, \omega, H)$ (with $\omega$ a symplectic form and $H$ a function on $M$) is called a (completely) integrable system if there exist $n$ functions $f_1, \ldots, f_n$ on $M$ which are a) constants of the motion for the dynamics defined by $H$, b) in involution (i.e., the Poisson bracket $\{f_i, f_j\}$ of each pair vanishes), c) functionally independent (i.e., their differentials $\dd f_i$ are linearly independent) on a dense subset of $M$. The Liouville--Arnol'd theorem states that an integrable system is foliated by Lagrangian submanifolds which, if they are compact and connected, are diffeomorphic to $n$-dimensional tori; the flow of each Hamiltonian vector field $X_{f_i}$ preserves the leaves of the foliation; and, in a neighbourhood of each torus, there are action-angle coordinates $(s_i, \varphi^i)$ such that
\begin{itemize}
    \item [a)] the symplectic form reads $\omega=\dd \varphi^i\wedge \dd s_i$, 
    \item [b)] the leaves of the foliation are level sets $\{s_1= c_1, \ldots, s_n = c_n\}$ of action coordinates, 
    \item [c)] the functions $f_i$ depend solely on the coordinates $(s_i)$, and thereupon their Hamiltonian vector fields are linear combinations of $\partial_{\varphi^i}$ with constant coefficients on each leave\footnote{Unfortunately, statement c) is frequently found in the literature as if it only held for $X_H$, but for a reader who is well acquainted with the Liouville-Arnol'd theorem and its proof it should be clear that it holds for every $X_{f_i}$. Refer to \cite[pp.~89-96]{Audin2004} for more details.} --namely
    $$f_i = f_i (s^1, \ldots, s^n)\, , \quad X_{f_i} = \sum_{j=1}^n \parder{f_i}{s_j} \parder{}{\varphi^j}\, .$$
\end{itemize}

In order to apply the Liouville--Arnol'd theorem \cite{Arnold1978}, one first needs to know $n$ independent conserved quantities in involution. A way to do so is to obtain a second Poisson structure $\Lambda_1$ which is compatible with the one defined by the original symplectic structure $\omega$. In that case, one can construct a $(1, 1)$-tensor field $N$, the so-called recursion operator, whose eigenvalues are functions in involution with respect to both Poisson brackets \cite{F.M.P2000,M.C.F+1997,M.M1984}. The theory of compatible Jacobi structures was developed by Iglesias, Monterde, Marrero, Nunes da Costa, Padrón and Petalidou \cite{I.M2001,M.M.P1999,NunesdaCosta1998a,P.N2003}. However, to the best of our knowledge, it has not been applied for studying the integrability of dynamics as in the Poisson case.

We have initiated a program of extending the theory of integrable systems to the realm of contact geometry. A natural way of doing so is regarding contact Hamiltonian systems as homogeneous symplectic ones. In \cite{C.d.L+2023} (see also \cite[Chapter 12]{Lopez-Gordon2024}), we proved a Liouville--Arnol'd theorem for contact Hamiltonian systems. A next natural step is attempting to extend the notion of bi-Hamiltonian system to this context.
Naively, one would think that, given a $(2n+1)$-dimensional manifold $M$ endowed with a contact form $\eta$ and a second Jacobi structure $(\Lambda_1, E_1)$ compatible with the one defined by $\eta$, their recursion operator could be used, under some non-degeneracy conditions, to construct a set of $(n+1)$ functions in involution.
However, it turns out that it is not possible to do this in the Jacobi category. Fernandes \cite{Fernandes1994} showed that a completely integrable system is bi-Hamiltonian and the recursion operator has the maximum possible functionally independent real eigenvalues if, and only if, the graph of the Hamiltonian function is a hypersurface of translation with respect to the affine structure defined by the action coordinates. This condition cannot hold for a non-trivial homogeneous Hamiltonian system, which implies that the recursion operator arising from two compatible Jacobi structures cannot produce a maximal set of functionally independent functions in involution. 
Nevertheless, bi-Hamiltonian structures can still be used for studying integrable contact systems by means of the symplectisation. We illustrate this in detail with an example.

\section{Preliminaries}

\subsection{Poisson and Jacobi manifolds}
Given a bivector field $\Lambda\in\Sec(\bigwedge^2 \T M)$ and a vector field $E\in \X(M)$ on a manifold $M$, let 
\begin{equation}\label{eq:Jacobi_bracket}
   \Cinfty(M)\times \Cinfty(M)\ni (f,g)\mapsto \{f, g\} = \Lambda(\dd f, \dd g) + f E(g) - g E(f)\in \Cinfty(M)\, .
\end{equation}
Lichnerowicz \cite{Lichnerowicz1977a,Lichnerowicz1978} proved that $\{\cdot, \cdot\}$ is a Lie bracket if, and only if, $[\Lambda, \Lambda]_{\mathrm{SN}} = 2 E \wedge \Lambda$ and $[E, \Lambda]_{\mathrm{SN}} = 0$, where $[\cdot, \cdot]_{\mathrm{SN}}$ denotes the Schouten--Nijenhuis bracket.
If these conditions hold, then $\{\cdot, \cdot\}$ (resp.~$(\Lambda, E)$) is called a \emph{Jacobi bracket} (resp.~\emph{Jacobi structure}) on $M$. The Jacobi bracket satisfies the \emph{weak Leibniz rule}:
\begin{equation}\label{eq:weak_Leibniz_rule}
    \{f, gh\} = \{f, g\} h + \{f, h\} g + g h E(f)\, ,\quad \forall\, f, g, h \in \Cinfty(M)\, .
\end{equation} 
A Poisson structure (resp.~bracket) is a Jacobi structure (resp.~bracket) with $E\equiv 0$. A collection of functions $f_1, \ldots, f_k\in \Cinfty(M)$ are said to be \emph{in involution} if $\{f_i, f_j\}=0$ for each $i, j\in \{1, \ldots, k\}$.

\subsection{Integrable contact systems}

A \emph{(co-oriented) contact manifold} $(M, \eta)$ is a $(2n+1)$-dimensional manifold endowed with a contact form $\eta$, i.e., a one-form such that $\eta \wedge \dd \eta^n$ is a volume form on $M$. The \emph{Reeb vector field} $\Reeb$ is the unique vector field on $\ker \dd \eta$ such that $\eta(\Reeb) = 1$. 
Any contact manifold can be covered by charts of Darboux coordinates $(q^i, p_i, z)$, in which $\eta=\dd z - p_i \dd q^i$ and $\Reeb = \partial_z$. The \emph{contact Hamiltonian vector field} $X_f$ of $f\in \Cinfty(M)$ is the unique vector field satisfying $\eta(X_f) = -f$ and $\liedv{X_f} \eta = -\Reeb(f) \eta$. In Darboux coordinates, 
$$X_f = \frac{\partial f}{\partial p_i} \frac{\partial}{\partial q^i} - \left( \frac{\partial f}{\partial q^i} + p_i \frac{\partial f}{\partial z}\right) \frac{\partial}{\partial p_i} + \left(  p_i \frac{\partial f}{\partial p_i} - f\right)\frac{\partial}{\partial z}\, . $$
A contact form $\eta$ on $M$ defines a Jacobi bracket $\{\cdot, \cdot\}$ by
$$\{f, g\} = X_f(g) + g \Reeb(f)\, ,\quad \forall\, f,g\in \Cinfty(M)\, .$$

\begin{definition}\label{def:integrable_systems}
    A \emph{contact Hamiltonian system} $(M, \eta, h)$ is a co-oriented contact manifold $(M, \eta)$ with a fixed Hamiltonian function $h\in \Cinfty(M)$. A \emph{dissipated quantity} $f\in \Cinfty{M}$ is a function in involution with $h$, i.e., $\{f, h\}=0$, where $\{\cdot, \cdot\}$ is the Jacobi bracket defined by $\eta$. 
    A $(2n+1)$-dimensional contact Hamiltonian system $(M, \eta, h)$ will be called a \emph{completely integrable contact system} if there exist $n+1$ dissipated quantities $f_1, \ldots, f_{n+1}$ on $M$ in involution ($\{f_i, f_j\}=0\ \forall\, i,j$) with $\operatorname{rank} \langle\dd f_1, \ldots, \dd f_{n+1}\rangle \geq n$.
\end{definition}

\subsection{Exact symplectic manifolds}


An \emph{exact symplectic manifold} is a pair $(M, \theta)$ formed by a manifold $M$ and a \emph{symplectic potential} $\theta$ on $M$, i.e., a one-form $\theta\in \Omega^1(M)$ such that $\omega = - \dd \theta$ is a symplectic form on $M$. The \emph{Liouville vector field} $\lvf\in \X(M)$ is the unique vector field such that $\contr{\lvf} \omega = - \theta$. A tensor field $A$ on $M$ will be called \emph{$k$-homogeneous} (or \emph{homogeneous of degree $k$}), where $k\in \ZZ$, if $\liedv{\lvf} A = k A$. In particular, both the symplectic potential $\theta$ and the symplectic form $\omega= - \dd \theta$ are $1$-homogeneous. Conversely, given a symplectic manifold $(M, \omega)$ endowed with a vector field $\lvf\in \X(M)$ such that $\liedv{\lvf} \omega = \omega$, it is straightforward to check that $\theta = - \contr{\lvf} \omega$ is a symplectic potential. For each function $f\in \Cinfty(M)$, its Hamiltonian vector field is determined by $\contr{X_f} \omega = \dd f$.

\begin{definition}\label{def:integrable_systems_homogeneous}
    A \emph{homogeneous Hamiltonian system} $(M, \theta, H)$ consists of a $2n$-dimensional exact symplectic  manifold $(M, \theta)$ and a $1$-homogeneous Hamiltonian function $H$. It is called a \emph{homogeneous integrable system} if there exist $n$ first integrals $f_1, \ldots, f_n$ of $X_H$, which are functionally independent, in involution and homogeneous of degree $1$ on a dense open subset $M_{0}\subseteq M$. 
\end{definition}

\subsection{Symplectisation of contact manifolds}

In this subsection, we briefly recall some concepts and results regarding the symplectisation of co-oriented contact manifolds. Refer to \cite{B.G.G2017a,C.d.L+2023,G.G2022a,I.L.M+1997,L.M1987,Lopez-Gordon2024} for more details.

Let $(M, \eta)$ be a co-oriented contact manifold, and let $(M^\Sigma, \theta)$ be an exact symplectic manifold.
A (locally trivial) fiber bundle $\Sigma\colon M^\Sigma \to M$ is a symplectisation if, and only if, there exists a nowhere-vanishing function $\sigma\colon M^\Sigma \to \RR$ such that
\begin{equation}
    \sigma \left(\Sigma^\ast \eta\right) = \theta\, .
\end{equation}
The function $\sigma$ is called the \emph{conformal factor} of $\Sigma$.

    Given a symplectisation $\Sigma\colon M^\Sigma \to M$ with conformal factor $\sigma$, there is a bijection $f\mapsto f^\Sigma = -\sigma \left(\Sigma^\ast f\right)$ between functions on $M$ and homogeneous functions of degree 1 on $M^\Sigma$ such that
    the Poisson and Jacobi brackets are related by 
    \begin{equation}
        \left\{f^\Sigma, g^\Sigma \right\}_\theta = \left\{f, g \right\}^\Sigma_\eta\, .
    \end{equation}
    If $(M, \eta, H)$ is a contact Hamiltonian system, a function $f$ on $M$ is a dissipated quantity with respect to $(M, \eta, H)$ if, and only if, $f^\Sigma$ is a conserved quantity with respect to $(M^\Sigma, \theta, H^\Sigma)$. Moreover, $(M, \eta, H)$ is an integrable contact system if, and only if, $(M^\Sigma, \theta, H^\Sigma)$ is a homogeneous integrable system.


\begin{example}[Trivial symplectisation]
    Let $(M, \eta)$ be a co-oriented contact manifold. Denote the positive real numbers by $\RR_+$ and by $r$ their canonical global coordinate. The trivial bundle $\Sigma\colon M\times \RR_+ \to M$ with the natural projection is a symplectisation with conformal factor $\sigma=r$. If $(q^i, p_i, z)$ are Darboux coordinates for $\eta$, we have that $\Sigma(q^i, p_i, z, r) = (q^i, p_i, z)$ and $\theta = r\dd z - rp_i\dd q^i$.
\end{example}

\section{Homogeneous bi-Hamiltonian systems}



Two Poisson tensors $\Lambda$ and $\Lambda_1$ on a manifold $M$ are said to be \emph{compatible} if $\Lambda + \Lambda_1$ is also a Poisson tensor on $M$. In other words, 
if $[\Lambda_1, \Lambda_2]_{\mathrm{SN}} = 0$.
Let $\Omega^1(M)\ni \alpha \mapsto \sharp_{\Lambda}(\alpha) = \Lambda(\cdot, \alpha) \in \X(M)$ denote the $\Cinfty(M)$-module morphism defined by $\Lambda$, and analogously for $\Lambda_1$.  A vector field $X\in \X(M)$ is called \emph{bi-Hamiltonian} if it is a Hamiltonian vector field with respect to two compatible Poisson structures, namely,
\begin{equation}
    X = \sharp_{\Lambda}(\dd h) = \sharp_{\Lambda_1}(\dd h_1 )\, ,
\end{equation}
for two functions $h, h_1 \in \Cinfty(M)$. If $\sharp_\Lambda$ is an isomorphism, then we can define the $(1, 1)$-tensor field 
$$N = \sharp_{\Lambda_1} \circ \sharp_{\Lambda}^{-1}\, ,$$
called the \emph{recursion operator}.
The pair $(\Lambda, N)$ is called a \emph{Poisson--Nijenhuis structure} on $M$, and the triple $(M, \Lambda, N)$ is called a \emph{Poisson--Nijenhuis manifold}. The eigenvalues $\lambda_i$ of $N$ are functions in involution with respect to the Poisson brackets defined by $\Lambda$ and $\Lambda_1$ \cite{M.C.F+1997}.

Consider a completely integrable Hamiltonian system $(M, \omega, H)$, where $M$ is $2n$-dimensional, with action-angle coordinates $(s_i, \varphi^i)$ satisfying the following conditions:

\begin{minipage}{0.95\linewidth}
\begin{itemize}
    \item [(ND)] The Hessian matrix $\left(\parderr{H}{s_i}{s_j}\right)$ of the Hamiltonian with respect to the action variables is non-degenerate in a dense subset of $M$.
    \item [(BH)] The system is bi-Hamiltonian and the recursion operator $N$ has $n$ functionally independent real eigenvalues $\lambda_1, \ldots, \lambda_n$.
\end{itemize}
\end{minipage}

Condition (ND), also known as the Kolmogorov condition, admits a coordinate-independent formulation; see \cite[Section 2]{Roy2006}. We omit it here, as it requires introducing geometric structures that play no further role in this work.

Fernandes \cite[Theorem 2.5]{Fernandes1994} showed that, under such assumptions, the Hamiltonian function can be written as
\begin{equation}\label{eq:H_lambda_i}
    H(\lambda_1, \ldots, \lambda_n) = \sum_{i=1}^n H_i(\lambda_i)\, ,
\end{equation}
where each $H_i$ is a function that depends only on the corresponding $\lambda_i$, for $i\in \{1, \ldots, n\}$. Moreover, the fact that $\{\lambda_1, \ldots, \lambda_n\}$ are in involution implies that $\{H_1, \ldots, H_n\}$ are in involution.

In our case, given a homogeneous integrable system (i.e., the symplectisation of an integrable contact system) satisfying (ND), we would like to find a second compatible Poisson structure which would satisfy (BH) and be homogeneous of degree $-1$. The latter would allow to project it into a Jacobi structure compatible with the original Jacobi structure\footnote{More precisely, given a Poisson manifold $(M, \Lambda)$ such that $\Lambda$ is $-1$-homogeneous with respect to $\lvf$, a $1$-codimensional submanifold $N$ of $M$ which is transverse to $\lvf$ can be equipped with a Jacobi structure (see \cite[Section 2]{D.L.M1991}) For our purposes, we can consider the trivial Poissonization/symplectisation --namely, $M=N\times \RR_+$ and $\lvf = r\partial_r$ with $r$ the canonical coordinate of $r$ (see the next section).}. Unfortunately, this is not possible.

\begin{proposition}\label{proposition:no-go}
    Let $(M, \theta, H)$ be a homogeneous integrable system satisfying the assumption (ND). Denote by $\Lambda$ the Poisson structure defined by $\omega = - \dd \theta$, and by $\lvf$ the Liouville vector field corresponding to $\theta$.
    For any Poisson structure $\Lambda_1$ on $M$ compatible with $\Lambda$, the following statements cannot be true simultaneously:
    \begin{enumerate}
        \item $\Lambda_1$ is $(-1)$-homogeneous (i.e., $\liedv{\lvf} \Lambda_1 = - \Lambda_1$) 
        \item $\Lambda_1$ satisfies (BH).
    \end{enumerate}
\end{proposition}

\begin{proof}
    Suppose that $\Lambda_1$ is $(-1)$-homogeneous. Then, $N$ is $0$-homogeneous, which in turn implies that its eigenvalues $\lambda_i$ are $0$-homogeneous. If $N$ has $n$ functionally independent eigenvalues $\lambda_1, \ldots, \lambda_n$, we can write $H$ in the form \eqref{eq:H_lambda_i}, but $H$ is $1$-homogeneous, so
    \begin{equation}
        H \equiv \lvf(H) = \sum_{i=1}^n H_i'(\lambda_i) \lvf(\lambda_i) \equiv 0\, .
    \end{equation}
    On the other hand, the assumption (ND) requires the Hessian matrix of $H$ to be non-singular, leading to a contradiction.
    $\hfill\square$
\end{proof}

Nevertheless, if $N$ is $1$-homogeneous and satisfies (BH), then its eigenvalues are $n$ functionally independent and $1$-homogeneous functions in involution, so they will project into $n$ functions in involution with respect to the Jacobi bracket.

\subsection{Example}\label{Example:Poisson}
Let $M=\RR^2$, and consider its cotangent bundle $\cT M \cong \RR^4$ endowed with the canonical one-form $\theta_{\RR^2}$. In bundle coordinates $(x^i, p_i)$, it reads $\theta_M = p_i \dd x^i$. It defines the symplectic form  $\omega_M = - \dd \theta_M = \dd x^i \wedge \dd p_i$, and the Poisson structure 
$$\Lambda = \parder{}{x^1} \wedge \parder{}{p_1} + \parder{}{x^2} \wedge \parder{}{p_2}\, .$$
In this case, the Liouville vector field is $\lvf_{M} = p_i \partial_{p_i}$, the infinitesimal generator of homotheties on the fibers.
A Poisson structure compatible with $\Lambda$ is 
$$\Lambda_1 =  p_1 \parder{}{x^1} \wedge \parder{}{p_1} + p_2 x^2 \parder{}{x^2} \wedge \parder{}{p_2}\, .$$
The Nijenhuis tensor $N = \sharp_{\Lambda_1} \circ \sharp_{\Lambda}^{-1}$ reads
$$ N =  p_1 \left( \parder{}{x^1} \otimes \dd x^1 +  \parder{}{p_1} \otimes \dd p_1 \right) + p_2 x^2 \left( \parder{}{x^2} \otimes \dd x^2 +  \parder{}{p_2} \otimes \dd p_2 \right)\, .$$
The eigenvalues of $N$ are $\lambda_1 = p_1$ and $\lambda_2 = p_2 x^2$, which are homogeneous of degree $1$, in involution with respect to both $\Lambda$ and $\Lambda_1$, and functionally independent on the dense subset $U = \cT M \setminus \Big(\{p_2 = 0\} \cap \{x^2=0\}\Big)$.
The vector field
$$X = \parder{}{x^1} + x^2 \parder{}{x^2} - p_2 \parder{}{p_2}$$ 
is bi-Hamiltonian. Indeed, it is the Hamiltonian vector field of $H=p_1+p_2x^2$ with respect to $\Lambda$, and the Hamiltonian vector field of $H_1=\log(p_1p_2 x^2)$ with respect to $\Lambda_1$. Moreover, $\lambda_1$ and $\lambda_2$ are first integrals of $X$.

Since $\lambda_1$ and $\lambda_2$ are functionally independent in $U$, they can be used as coordinates. We can take the $0$-homogeneous functions $\varphi^1 = x^1$ and $\varphi^2 = \log x^2$, so that $\theta = \lambda_i \dd \varphi^i$.
The coordinates $\lambda_1$ and $\lambda_2$ are the action coordinates, while $\varphi^1$ and $\varphi^2$ are angle coordinates. In these coordinates,
$$\Lambda = \sum_{i=1}^2 \parder{}{\varphi^i} \wedge \parder{}{\lambda^i}\, ,\quad \Lambda_1 = \sum_{i=1}^2 \lambda_i \parder{}{\varphi^i} \wedge \parder{}{\lambda^i}\, ,\quad  X = \sum_{i=1}^2 \parder{}{\varphi^i}\, ,$$
and $H=\lambda_1 + \lambda_2$.



\section{Integrable contact systems}

Given a Jacobi structure $(\Lambda, E)$ on $M$, one can construct an associated Poisson structure $\tilde{\Lambda} = \frac{1}{r} \Lambda + \partial_r \wedge E$ on $M\times \RR_+$, which by construction is homogeneous of degree $-1$ with respect to $\lvf = r \partial_r$.
Two Jacobi structures $(\Lambda, E)$ and $(\Lambda_1, E_1)$ on a manifold $M$ are called compatible if $(\Lambda+ \Lambda_1, E+E_1)$ is also a Jacobi structure on $M$. Nunes da Costa \cite[Proposition 1.7]{NunesdaCosta1998a} showed that $(\Lambda, E)$ and $(\Lambda_1, E_1)$ are compatible Jacobi structures if, and only if, $\tilde{\Lambda}$ and $\tilde{\Lambda}_1$ are compatible Poisson structures.

A consequence of Proposition~\ref{proposition:no-go} is the following:

\begin{corollary}\label{corollary:main}
    Let $(M, \eta, H)$ be a $(2n+1$)-dimensional integrable contact system satisfying the assumption (ND). For any Jacobi structure $(\Lambda_1, E_1)$ compatible with the Jacobi structure $(\Lambda, E)$ defined by $\eta$, 
    the recursion operator $N=\sharp_{\tilde{\Lambda}_1} \circ\, \sharp_{\tilde{\Lambda}}^{-1}$ relating the associated Poisson structures on $M\times \RR_+$ cannot have $(n+1)$ functionally independent real eigenvalues.
\end{corollary}

Consequently, compatible Jacobi structures cannot be utilised to construct a set of independent dissipated quantities in involution for a contact Hamiltonian system. Nevertheless, we can symplectise the contact Hamiltonian system and obtain a second Poisson structure compatible with the one defined by the exact symplectic structure. 

\subsection{Example}\label{Example:contact}
Consider the contact Hamiltonian system $(M=\RR^3, \eta, h)$, with $\eta$ the canonical contact form, $\eta = \dd z - p \dd q$, and $h = p - z$. The contact Hamiltonian vector field of $h$ is $X_h=\partial_q + p \partial_p + z\partial_z$. The symplectisation of $(M, \eta, h)$ is $(\RR^4, \theta, H)$, with
$$ \theta = r \dd z - rp \dd q\, , \quad H = rz - rp\, ,$$
and Liouville vector field $\lvf = r\partial_r$, in canonical bundle coordinates $(q, p, z, r)$. If we make the change of coordinates
$x^1 = q, \, x^2 = z, \, p_1 = -rp, \,p_2 = r$,
we have that
$$\theta = p_i \dd x^i\, , \quad H = p_1 + p_2 x^2\, \quad \lvf = p_i \parder{}{p_i}\, , \quad i\in\{1, 2\}\, .$$
This is precisely the system from Example~\ref{Example:Poisson}. Thus, we have the functions $\lambda_1 = p_1 = -rp$ and $\lambda_2 = p_2 x^2 = rz$, which are homogeneous of degree $1$, in involution, and functionally independent on a dense subset. Projecting them to $M$, we obtain $\bar{\lambda}_1 = p$ and $\bar{\lambda}_2 = -z$, which are functionally independent dissipated quantities in involution.

Moreover, the angle coordinates $\varphi^1 = x^1 = q$ and $\varphi^2 = \log x^2 = \log z$ are $0$-homogeneous, so they project into $M$. With a slight abuse of notation, we will also denote by $\varphi^1$ and $\varphi^2$ to the corresponding functions on $M$. Let $\bar{\lambda} = -\bar{\lambda}_1/\bar{\lambda}_2 = p/z$. In the chart $(U = M \setminus\{z=0\};\varphi^1, \varphi^2, \bar{\lambda})$, the contact Hamiltonian vector field reads $X_h = \partial_{\varphi^1} + \partial_{\varphi^2}$. Moreover, $\bar{\eta} = \dd \varphi^2 -  \bar{\lambda} \dd \varphi^1$ is a contact form on $U$ conformal to $\eta$ (i.e., they generate the same contact distribution $\ker \bar \eta = \ker\eta$), and $X_h$ is the Hamiltonian vector field of $\bar{h}=\bar{\lambda} - 1$ with respect to $\bar{\eta}$.

\section{Conclusions and outlook}

We have proven the impossibility of utilising compatible Jacobi structures to construct completely integrable contact systems (see Corollary~\ref{corollary:main}). This is a consequence of the very restrictive character of bi-Hamiltonian structures. We would like to explore less restrictive geometric structures for studying the integrability contact systems.

In the last years, Tempesta, Tondo and their collaborators have developed a theory on Haantjes operators (a generalisation of Nijenhuis operators). In particular, they have shown that a Hamiltonian system is completely integrable if and only if it has an associated symplectic--Haantjes structure (see \cite[Section 4]{T.T2022}). In subsequent works, we plan to study the applicability of Haantjes structures for the study of homogeneous integrable systems (which are equivalent to integrable contact systems).

On the other hand, in this paper and in our previous article \cite{C.d.L+2023}, we consider functions in involution, which means that their Hamiltonian vector fields (that are assumed to be complete) generate an action of an Abelian Lie group on the manifold. Instead, we could consider a non-Abelian Lie subalgebra on the Poisson algebra of functions, in such a way that the isomorphic\footnote{Recall that $\Cinfty(M)\ni f \mapsto X_f \in \X(M)$ is an (anti-)isomorphism between the Lie algebras $(\Cinfty(M), \{\cdot, \cdot\})$ and $(\X(M), [\cdot, \cdot])$. Whether it is an isomorphism or an anti-isomorphism is just a question of sign criteria in the definition of Poisson bracket.} Lie algebra of Hamiltonian vector fields would generate a non-Abelian action on the manifold. This is related with the notion of contact Lie systems introduced by de Lucas and Rivas in \cite{d.R2023}.

Furthermore, we plan to explore the applicability of our theory on relevant physical examples that can be modelled by homogeneous integrable systems or contact integrable systems.

\begin{credits}
\subsubsection{\ackname} The authors wish to thank Janusz Grabowski for pointing out a logical imprecision in a previous version of this work. They are also grateful to the anonymous referees for their constructive comments. L.~Colombo and M.~de León received financial support from Grant PID2022-137909NB-C21 funded by MICIU/AEI/10.13039/501100011033. M.~de León also acknowledges financial support from Grant CEX2023-001347-S funded by MICIU/AEI/10.13039/501100011033.

\subsubsection{\discintname}
The authors have no competing interests to declare.
\end{credits}


\let\emph\oldemph

\printbibliography

\end{document}